\documentclass{amsart}

\usepackage{graphicx}
\usepackage{latexsym,color}
\usepackage{graphicx}
\usepackage{amssymb}
 \usepackage{amsfonts}
 \usepackage{amsmath}
\usepackage{amsthm}

\input{style}

\begin{document}

\title{Extended Weak Convergence and Utility Maximization with Proportional Transaction Costs}

 \author{Erhan Bayraktar \address{Department of Mathematics, University of Michigan} \email{erhan@umich.edu}} \thanks{}
 \author{Leonid Dolinskyi
 \address{the University of the State Fiscal Service, Department of Economic Cybernetics}\email{dolinsky.com@gmail.com}}
\author{Yan Dolinsky  \address{
  Department of Statistics, Hebrew University} \email{yan.dolinsky@mail.huji.ac.il}} \thanks{Y. Dolinsky is supported by the GIF Grant 1489-304.6/2019 and the ISF grant 160/17.}

\date{}


\maketitle \markboth{}{}
\renewcommand{\theequation}{\arabic{section}.\arabic{equation}}
\pagenumbering{arabic}

\begin{abstract}
In this paper we
study utility maximization with proportional transaction costs.
Assuming extended weak convergence of the underlying processes
we prove the convergence
of the corresponding utility maximization problems.
Moreover, we establish a limit theorem for the optimal trading strategies.
The proofs are based on the extended weak convergence theory developed in \cite{A} and the Meyer--Zheng topology
introduced in \cite{MZ}.
\end{abstract}
\tableofcontents
\section{Introduction}
\label{intro}
We deal with
the continuity of the utility maximization problem in the presence of proportional transaction costs,
under convergence in distribution of the financial markets.
More specifically, we focus on utility maximization from terminal wealth under
an admissibility condition on the wealth processes.
Although the problem of utility maximization with proportional transaction costs
was widely studied
(see, for instance,
\cite{MC:76, TKA:88, DN:90,DL:91, SS:94, CK:96, DPT:01, G:02, Bruno:02, KK:10, CO:11, GGKS:14, CS:16, BY:18}),
to the best of our knowledge, the continuity under weak convergence
was not considered before.

Clearly, the problem of utility maximization depends
on the flow of information (filtration). Hence, one should not expect that convergence of asset prices
alone will imply the convergence of the corresponding control problems.
In particular, there are many examples (see for instance page 75 in \cite{A}) of processes which are ``close" to each
other in distribution but the behaviour of the corresponding filtrations is completely different. This brings us to
a stronger
notion of convergence.

In his important unpublished manuscript \cite{A}, Aldous
introduced the notion of extended weak convergence
and showed that for optimal stopping this is the right notion of convergence.
Extended weak convergence is defined via weak convergence of
the corresponding prediction processes. The prediction process is a
measure valued process representing the regular
conditional distributions of the original stochastic process, and so it captures the structure of the information flow.
It is important to mention the recent papers \cite{BBBE1,BBBE2}
which provide a novel approach in the discrete time setup to
weak convergence topologies that take the information flow into account.

In this work we consider a sequence of continuous time
financial markets
with continuous asset prices
which converge to a limit market.
To ease notations we focus on the case of one risky asset.
Our main assumptions are extended weak convergence of the underlying processes and an existence of
strictly consistent price systems.
Under these natural conditions we prove that for a continuous and concave utility function,
there is a convergence of expected utilities; see Theorem~\ref{thm2.1}.
Moreover, we obtain a limit theorem for the optimal trading strategies; see Theorem~\ref{thm2.2}.
This paper is the first work that applies extended weak convergence to continuous time portfolio optimization problems.

In addition to the extended weak convergence, we also apply the
Meyer--Zheng topology which was introduced in
\cite{MZ}. Roughly speaking, the Meyer--Zheng
topology on the set of functions
is the topology of convergence in measure. As we will see,
this topology perfectly fits
for hedging with proportional transaction costs.
More precisely, the admissibility condition
will imply tightness of the trading strategies
in the Meyer--Zheng topology.

The current work is a continuation of \cite{BDG} where a similar problem was studied in
the frictionless setup (i.e. with no transaction costs).
Surprisingly, for the frictionless setup extended weak convergence is not a sufficient assumption,
and in order to
have convergence for the expected utilities one needs to require convergence of the
equivalent martingale measures (see Assumption 2.5 in \cite{BDG}).
These objects can be viewed as consistent price
systems for the frictionless case.
It is important to emphasize that with the presence of proportional transaction costs there is no need to assume any convergence
structure on the consistent price systems.
Namely, the presence of proportional
transaction costs provides the needed compactness.

The rest of the paper is organized as follows. In the next section we introduce
the setup and list our assumptions. In Section \ref{sec:2}, we formulate our main results, which are proved
in Section \ref{sec:3}.
In Section~\ref{sec:5}, we provide a specific example for financial markets which converge to a stochastic volatility model.
We show that in the presence of proportional transaction costs the expected utilities converge,
while in the frictionless setup there is no convergence.

\section{Preliminaries and Assumptions}\label{sec:1}
\subsection{Hedging with Proportional Transaction Costs}
We consider a model with one risky asset which we denote by
$S=(S_t)_{0\leq t\leq T}$, where $T<\infty$ is the time horizon.
We assume that the investor has a bank account that,
for simplicity, bears no interest. The process $S$ is assumed to be an adapted,
strictly positive and continuous process (not necessarily a semi-martingale) defined on a filtered probability
space $(\Omega,\mathcal F,(\mathcal F_t)_{0\leq t\leq T},\mathbb P)$ where
the filtration $(\mathcal F_t)_{0\leq t\leq T}$ satisfies the usual assumptions (right continuity and completeness).

Let $\kappa\in (0,1)$ be a constant. Consider a model in which every purchase or sale
of the risky asset at time $t\in [0,T]$ is subject to a proportional transaction
cost of rate $\kappa$.
A trading strategy is an adapted process
$\gamma=(\gamma_t)_{0\leq t\leq T}$
of
bounded variation with right continuous paths.
The random variable $\gamma_t$ denotes
the number of shares at time $t$.
We use the convention $\gamma_{0-}=0$. Moreover, we require that $\gamma_T=0$ which means that we
liquidate the portfolio at the maturity date.

Let $\gamma_t:=\gamma^{+}_t-\gamma^{-}_t$, $t\in [0,T]$
be the Jordan
decomposition into a positive variation
$(\gamma^{+}_t)_{0\leq t\leq T}$ and a negative variation
$(\gamma^{-}_t)_{0\leq t\leq T}$.
Since the bid price process is $(1-\kappa) S$ and the ask price process is $(1+\kappa) S$, then the portfolio value of a trading strategy
$\gamma$ at time $t$
is given by
$$V^{\gamma}_t:=(1-\kappa)\int_{0}^t S_ud\gamma^{-}_u-(1+\kappa) \int_{0}^t S_ud\gamma^{+}_u+(1-\kappa) S_t(\gamma_t)^{+}-(1+\kappa)S_t(\gamma_t)^{-}
$$
where all the
above integrals are pathwise Riemann--Stieltjes integrals.
We notice that the integrals take into account the possible transaction at $t=0$.
By rearranging the terms we get:
\begin{equation}\label{2.0}
V^{\gamma}_t=
\gamma_t S_t-\int_{0}^t S_u d\gamma_u-\kappa |\gamma_t|S_t-\kappa\int_{0}^t S_u|d\gamma_u|, \ \ t\in [0,T].
\end{equation}

Observe that the wealth process $(V^{\gamma}_t)_{0\leq t\leq T}$
is RCLL (right continuous with left limits) and from the fact that $\gamma_T=0$ it follows that
$V^{\gamma}_{T-}=V^{\gamma}_T$.
For any initial capital $x>0$ we denote by $\mathcal A(x)$
the set of all trading strategies $\gamma$ which satisfy the admissibility condition
$x+V^{\gamma}_t\geq 0$, for all $t\in [0,T]$.

Next, we introduce our utility maximization problem.
Let $C[0,T]$ be the space
of all continuous functions
$f:[0,T]\rightarrow\mathbb R$ equipped with the uniform topology.
Consider a continuous function $U:(0,\infty) \times C[0,T]\rightarrow\mathbb R$ such that
for any $s\in C[0,T]$ the function $U(\cdot,s)$ is non--decreasing and concave.
We extend $U$ to $\mathbb R_{+}\times C[0,T]$ by
$U(0,s):=\lim_{v\downarrow 0} U(v,s)$ (the limit might be $-\infty$).

\begin{asm}\label{asm2.1}
${}$\\
(i) For any $x>0$ we have
$\mathbb E_{\mathbb P}[U(x,S)]>-\infty$.\\
(ii) There exist continuous functions
$m_1,m_2:[0,1)\rightarrow\mathbb{R}_{+}$ with $m_1(0)=m_2(0)=0$ (modulus of continuity) and
a non-negative random variable $\zeta\in L^1(\Omega,\mathcal F,\mathbb P)$ such that for any $\lambda\in (0,1)$ and
$v>0$
$$
U((1-\lambda) v,S)\geq (1-m_1(\lambda)) U(v,S)-m_2(\lambda) \zeta.
$$
 \end{asm}

For a given initial capital $x>0$ consider the optimization problem
\begin{equation*}
u(x):=\sup_{\gamma\in\mathcal A(x)}\mathbb E_{\mathbb P}[U(x+V^{\gamma}_T,S)],
\end{equation*}
where, for a random variable $X$ which satisfies $\mathbb E_{\mathbb P}[\max(-X,0)]=\infty$ we set
$\mathbb E_{\mathbb P}[X]:=-\infty$.
Let us notice that Assumption \ref{asm2.1}(i) implies $u(x)>-\infty$.
\begin{rem}\label{rem.short}
We should note that the power and the log utility satisfy Assumption \ref{asm2.1}. Moreover, for a continuous function
$g:\mathbb R\rightarrow\mathbb R_{+}$
the utility $U(v,s):=-(g(s_T)-v)^{+}$ which corresponds to shortfall risk minimization for a vanilla option with the payoff
$g(S_T)$, also satisfies Assumption \ref{asm2.1} (provided that $\mathbb E_{\mathbb P}[g(S_T)]<\infty$).
Indeed, in this case,
if $v\geq \frac{g(S_T)}{1-\lambda}$ then
$U((1-\lambda) v,S)=U(v,S)=0.$ If $v<\frac{g(S_T)}{1-\lambda}$, then
$$|U((1-\lambda) v,S)-U(v,S)|\leq \lambda v\leq \frac{\lambda}{1-\lambda} g(S_T).$$
Thus, for
$m_1(\lambda):= 0$, $m_2(\lambda):=\frac{\lambda}{1-\lambda}$ and $\zeta:=g(S_T)$, Assumption \ref{asm2.1} holds true.
\end{rem}
\subsection{Approximating Sequence of Models}
For any $n$,
let $S^{(n)}=(S^{(n)}_t)_{0\leq t\leq T}$.
 be a strictly positive, continuous process defined on some filtered probability space
$(\Omega_n,\mathcal F^{(n)},(\mathcal F^{(n)}_t)_{0\leq t\leq T},\mathbb P_n)$, the filtration
$(\mathcal F^{(n)}_t)_{0\leq t\leq T}$ satisfies the usual assumptions .
For the $n$--th model, a trading strategy
is a right continuous adapted processes
$\gamma^{(n)}=(\gamma^{(n)}_t)_{0\leq t\leq T}$
of bounded variation which satisfies $\gamma^{(n)}_T=0$. As before, we use the convention that
$\gamma^{(n)}_{0-}=0$. The corresponding portfolio value is
given by
\begin{equation*}
V^{\gamma^{(n)}}_t:=\gamma^{(n)}_tS^{(n)}_t-\int_{0}^t S^{(n)}_ud\gamma^{(n)}_u-\kappa |\gamma^{(n)}_t|S^{(n)}_t -\kappa\int_{0}^t S^{(n)}_u|d\gamma^{(n)}_u|, \ \ t\in [0,T].
\end{equation*}
For any $x>0$ we denote by $\mathcal A^{(n)}(x)$ the set of all trading strategies $\gamma^{(n)}$ which satisfy
$x+V^{\gamma^{(n)}}_t\geq 0$, for all $t\in [0,T]$.
Set $$u_n(x):=\sup_{\gamma^{(n)}\in\mathcal A^{(n)}(x)}\mathbb E_{\mathbb P_n}\left[U\left(x+V^{\gamma^{(n)}}_T,S^{(n)}\right)\right].$$

The following assumption will be essential for proving tightness and some integrability properties of the admissible trading strategies.
\begin{asm}\label{asm2.2}
There exists $\varepsilon\in(0,\kappa)$
and probability measures $\mathbb Q\sim\mathbb P$, $\mathbb Q_n\sim\mathbb P_n$, $n\in\mathbb N$
with the following properties:\\
1)There exists $\mathbb Q$ local--martingale $(M_t)_{0\leq t\leq T}$
and for any $n\in\mathbb N$ there exists a $\mathbb Q_n$ local--martingale $(M^{(n)}_t)_{0\leq t\leq T}$ such that
\begin{equation*}
|M-S|\leq (\kappa-\varepsilon)S, \ \ |M^{(n)}-S^{(n)}|\leq (\kappa-\varepsilon)S^{(n)}, \ \ \forall n\in\mathbb N.
\end{equation*}
2) The sequence of probability measures $\mathbb P_n$, $n\in\mathbb N$, is
contiguous to the sequence $\mathbb Q_n$, $n\in\mathbb N$. Namely,
for any sequence of events $A_n\in\mathcal F^{(n)}$, $n\in\mathbb N$
if $\lim_{n\rightarrow\infty} \mathbb Q_n(A_n)=0$ then
$\lim_{n\rightarrow\infty} \mathbb P_n(A_n)=0$.
\end{asm}
\begin{rem}
Assumption \ref{asm2.2} can be viewed as a uniform version of
 the existence of strictly consistent price systems.
For a given model the existence of a strictly consistent price system is equivalent to
robust no free lunch with vanishing risk condition for simple
strategies; see e.g. \cite{GLR}.
\end{rem}

Next, we assume the
following uniform integrability assumptions.
\begin{asm}\label{asm2.2+}
${}$\\
(i)  For any $x>0$
the set $\{U^{-}(x,S^{(n)})\}_{n\in\mathbb N}$ is uniformly integrable
where $U^{-}:=\max(-U,0)$. \\
(ii) For any $x>0$ the set $\left\{U^{+}
\left(x+V^{\gamma^{(n)}}_T,S^{(n)}\right)\right\}_{n\in\mathbb N,
\gamma^{(n)}\in\mathcal A^{(n)}(x)}$ is uniformly integrable, where
$U^{+}:=\max(U,0)$.
\end{asm}
Notice, that from Assumption \ref{asm2.2+}
\begin{equation*}\label{2.finite}
-\infty<\lim\inf_{n\rightarrow\infty} u_n(x)\leq \lim\sup_{n\rightarrow\infty} u_n(x)<\infty, \ \ \forall x>0.
\end{equation*}
In general, if $U$ is not bounded from above, then the verification of Assumption \ref{asm2.2+}(ii) can be a difficult task.
The following result gives quite general and easily verifiable conditions
which are sufficient for Assumption~\ref{asm2.2+} to hold true.
\begin{prop}\label{prop1}
Suppose there exist constants $L>0$, $0<\alpha<1$ and $q>\frac{1}{1-\alpha}$ which satisfy the following. \\
(I) For all $(v,s)\in (0,\infty) \times C[0,T]$,
$U(v,s)\leq L (1+v^{\alpha}).$\\
(II) For any $n\in\mathbb N$ there exists
a probability measure $\hat {\mathbb Q}_n\sim \mathbb P_n$, a $\hat{\mathbb Q}_n$ local--martingale $(\hat M^{(n)}_t)_{0\leq t\leq T}$ such that
$|\hat M^{(n)}-S^{(n)}|\leq \kappa S^{(n)}$ (i.e. $(\hat M^{(n)},\hat {\mathbb Q}_n)$ is a consistent price system)
and
$
\sup_{n\in\mathbb N}\mathbb E_{{{\hat{\mathbb Q}}}_n}\left[\left(\frac{d{\mathbb P}_n}{d{\hat{\mathbb Q}}_n}\right)^{q}\right]<\infty.
$\\
Then Assumption \ref{asm2.2+}(ii) holds true.
\end{prop}
\begin{proof}
The proof is done by using similar arguments as in Lemma 2.2 in \cite{BDG}. The only needed property is that
for any $n\in\mathbb N$ and $\gamma^{(n)}\in\mathcal A_n(x)$
we have $\mathbb E_{\hat{\mathbb Q}_n}[V^{\gamma^{(n)}}_T]\leq  0$.
This is a well known property of consistent price systems (see, for instance, Proposition 1.6 in \cite{S}).
\end{proof}
We end this section with an example of market models which satisfy Assumptions \ref{asm2.2}--\ref{asm2.2+}.
\begin{exm}
Assume that for any $n$ the process $(S^{(n)}_t)_{0\geq t\leq T}$ is given by the SDE
\begin{equation}\label{SDE}
dS^{(n)}_t=S^{(n)}_t\left(\mu^{(n)}_t dt+\sigma^{(n)}_t dW^{(n)}_t\right), \ \ t\in [0,T]
\end{equation}
where $W^{(n)}$ is a Brownian motion and the processes $\mu^{(n)},\sigma^{(n)}$ are predictable, with respect to $\mathcal F^{(n)}$.
Similarly, assume that
\begin{equation*}
dS_t=S_t\left(\mu_t dt+\sigma_t dW_t\right), \ \ t\in [0,T]
\end{equation*}
where $W$ is a Brownian motion and the processes $\mu,\sigma$ are predictable, with respect to $\mathcal F$.
Moreover, we assume that $\sigma,\sigma^{(n)}$ are invertible and there exists a constant $\mathcal C$ (which does not depend on $n$) such that
\begin{equation}\label{example}
\sup_{0 \leq t\leq  T}\left|\frac{\mu_t}{\sigma_t}\right |,\sup_{0 \leq t\leq  T}\left|\frac{\mu^{(n)}_t}{\sigma^{(n)}_t}\right |< \mathcal C \  \mbox{a.s} \ \ \forall{n\in\mathbb N}.
\end{equation}
Then,
from the Girsanov theorem it follows that
Assumption \ref{asm2.2}(1) holds true for
$M:=S$, $M^{(n)}:=S^{(n)}$, $n\in\mathbb N$  and
\begin{eqnarray*}
&\frac{d\mathbb Q}{d\mathbb P}:=\exp\left(-\int_{0}^T\frac{\mu_t}{\sigma_t}dW_t-
\int_{0}^T\frac{1}{2}\left(\frac{\mu_t}{\sigma_t}\right)^2 dt\right), \\
&\frac{d\mathbb Q_n}{d\mathbb P_n}:=\exp\left(-\int_{0}^T\frac{\mu^{(n)}_t}{\sigma^{(n)}_t}dW^{(n)}_t-
\int_{0}^T\frac{1}{2}\left(\frac{\mu^{(n)}_t}{\sigma^{(n)}_t}\right)^2 dt\right), \ \ n\in\mathbb N.
\end{eqnarray*}
From (\ref{example}) we have
$$\sup_{n\in\mathbb N}\mathbb E_{{{{\mathbb Q}}}_n}\left[\left(\frac{d{\mathbb P}_n}{d{{\mathbb Q}}_n}\right)^{q}\right]<\infty , \ \ \forall q\in\mathbb R,$$
and so Assumption \ref{asm2.2}(2) holds true.
Moreover, in view of Proposition \ref{prop1}, it follows that Assumption \ref{asm2.2+} holds true
provided that there exists $L,\alpha>0$ such that
$$U(v,s)\leq L (1+v^{\alpha}), \ \ \forall (v,s)\in (0,\infty) \times C[0,T].$$

Another type of example can be obtained by linearly
interpolating discrete time processes, which are discrete time analogs of (\ref{SDE}).
In Section \ref{sec:5} we provide a detailed analysis of a particular example of this type.
\end{exm}

\subsection{Extended Weak Convergence}
We start with formulating
our convergence assumptions.
\begin{asm}\label{asm2.3}
For any $k\in\mathbb N$ let $\mathbb D([0,T];\mathbb R^k)$ be the space
of all RCLL functions
$f:[0,T]\rightarrow\mathbb R^k$ equipped with the Skorokhod topology (for details see \cite{B}).
We assume that there exists $m\in\mathbb N$ and a stochastic processes
$X^{(n)}:\Omega_n\rightarrow \mathbb D([0,T];\mathbb R^m)$, $n\in\mathbb N$,
$X:\Omega\rightarrow C([0,T];\mathbb R^m)$ (i.e. $X$ is continuous) which satisfy the following:\\
(i) The filtrations $(\mathcal F^{(n)}_t)_{0\leq t\leq T}$, $n\in\mathbb N$ and $(\mathcal F_t)_{0\leq t\leq T}$, are the usual filtrations generated
by $X^{(n)}$, $n\in\mathbb N$ and $X$, respectively. \\
(ii) We have the weak convergence
\begin{equation*}
(S^{(n)}, X^{(n)})\Rightarrow (S,X) \ \ \mbox{on} \ \ \mathbb D([0,T];\mathbb R^{m+1}).
\end{equation*}
(iii) We have the extended weak convergence
$X^{(n)}\Rrightarrow X$. This means that for any $k$ and a
continuous bounded function $\psi: \mathbb D([0,T];\mathbb R^m)\rightarrow\mathbb R^k$
we have
\begin{equation*}
(X^{(n)},Y^{(n)})\Rightarrow (X,Y) \ \ \mbox{on} \ \ \mathbb D([0,T];\mathbb R^{m+k}),
\end{equation*}
where
$$Y^{(n)}_t:=\mathbb E_{\mathbb P_n}\left(\psi(X^{(n)})|\mathcal F^{(n)}_t\right) \ \ \mbox{and} \ \
Y_t:=\mathbb E_{\mathbb P}\left(\psi(X)|\mathcal F_t\right), \ \ t\in [0,T].$$
\end{asm}
\begin{rem}\label{extend}
In \cite{A} Aldous introduced the notion of “extended weak convergence” via prediction processes.
He proved that extended weak convergence
is equivalent to a more elementary condition which does
not require the use of prediction processes (see \cite{A}, Proposition 16.15).
This is the definition that we use above.

The verification of extended weak convergence was studied in \cite{A} and \cite{JS}.
Citing Aldous (page 130 in \cite{A}) ``any weak convergence result proved by the martingale
technique can be improved to extended weak convergence". In particular if the processes $X^{(n)}$, $n\in\mathbb N$ have independent increments
and $X$ is continuous in probability then the weak convergence
$X^{(n)}\Rightarrow X$ implies the extended weak convergence
$X^{(n)}\Rrightarrow X$ (see Proposition 20.18 in \cite{A} and Corollary 2 in \cite{JS}). For more results related to extended weak convergence
see \cite{JS}.
\end{rem}
\section{Main Results}\label{sec:2}

We are ready to state our first limit theorem.
\begin{thm}\label{thm2.1}
Under Assumptions \ref{asm2.1}-\ref{asm2.3} we have that
\begin{equation*}
u(x)=\lim_{n\rightarrow\infty} u_n(x),
\end{equation*}
for any $x>0$.
\end{thm}

A natural question is whether we have some kind of convergence
for the optimal trading strategies (optimal controls) as well.
In order to formulate our limit theorem for the
optimal controls we need some preparations.

Any function $f\in \mathbb D[0,T]:=\mathbb D([0,T];\mathbb R)$ can be extended to a function
$f:\mathbb R_{+}\rightarrow \mathbb R$ by $f(t):=f(T)$ for all $t\geq T$.
The Meyer--Zheng topology, introduced in \cite{MZ}, is a relative topology, on the
image measures on graphs
$(t,f(t))$
of trajectories $t\rightarrow f(t)$, $t\in\mathbb R_{+}$ under the measure
$\lambda(dt):=e^{-t} dt$
(called pseudo-paths), induced by the weak topology on probability laws on
the compactified space
$[0,\infty]\times \overline{\mathbb R}$.
From Lemma 1 in \cite{MZ} it follows that the Meyer--Zheng topology on the space $\mathbb D[0,T]$
is given by the
metric
$$d_{MZ}(f,g):=\int_{0}^T \min(1,|f(t)-g(t)|)dt+|f(T)-g(T)|, \ \ f,g\in \mathbb D[0,T].$$
We denote the corresponding space by $\mathbb D_{MZ}[0,T]$.
\begin{rem}
In Lemma 1 in \cite{MZ} the authors proved that the Meyer--Zheng topology on the space
$\mathbb D[0,\infty)$ is equivalent to convergence in measure. Since in our setup the functions are constants on the time interval
$[T,\infty)$ then convergence in measure is given by the above $d_{MZ}$ metric.
\end{rem}
Now we are ready to formulate our second limit theorem.
\begin{thm}\label{thm2.2} Assume that Assumptions \ref{asm2.1}-\ref{asm2.3} hold.
 Let $x>0$ and $\hat\gamma^{(n)}\in\mathcal{A}^{(n)}(x)$, $n\in\mathbb N$ be a sequence of asymptotically optimal portfolios, namely
\begin{equation}\label{optimal}
\lim_{n\rightarrow\infty} \left(u_n(x)-\mathbb E_{\mathbb P_n}\left[U\left(x+V^{\hat\gamma^{(n)}}_T,S^{(n)}\right)\right]\right)=0.
\end{equation}
Then the sequence (of laws)
$(S^{(n)},X^{(n)},\hat\gamma^{(n)})$,
 $n\in\mathbb N$ is tight on the space
 $\mathbb D([0,T];\mathbb R^{1+m})\times \mathbb D_{MZ}[0,T]$, and thanks to Prohorov's theorem (see \cite{B}), it is relatively compact.
 Moreover,
any cluster point of the sequence $(S^{(n)},X^{(n)},\hat\gamma^{(n)})$, $n\in\mathbb N$
(there is at least one)
is of the form $(S,X,\hat\gamma)$.
Define
\begin{equation}\label{portfolio}
\hat\gamma^{\mathcal F}_t:=\mathbb E_{\mathbb P}\left(\hat{\gamma}_t|\mathcal F_t\right), \ \ t\in [0,T]
\end{equation}
where, as before $(\mathcal F_t)_{0\leq t\leq T}$
is the usual filtration generated by $X$ and (with abuse of notations)
$\mathbb E_{\mathbb P}$  denotes the expectation on the corresponding probability space.
 Then, $\hat\gamma^{\mathcal F}=(\hat\gamma^{\mathcal F}_t)_{0\leq t\leq T}$ is an optimal portfolio, i.e.  $\hat\gamma^{\mathcal F}\in\mathcal A(x)$ and
 $$u(x)=\mathbb E_{\mathbb P}\left[U\left(x+V^{\hat\gamma^{\mathcal F}}_T,S\right)\right].$$
 \end{thm}
 We end this section with the following remark about Theorem \ref{thm2.2}.
\begin{rem}
 In view of Assumption \ref{asm2.3}(ii) any cluster point of the
 sequence $(S^{(n)},X^{(n)},\hat\gamma^{(n)})$,
 $n\in\mathbb N$
 has to be of the form $(S,X,\hat\gamma)$.
We can show
that for a such cluster point $(S,X,\hat\gamma)$, $\hat\gamma$ is an optimal portfolio.
However, $\hat\gamma$ is not necessarily adapted to the filtration
generated by $X$.

One possible way to treat this issue is to follow the weak formulation setup.
Roughly speaking, the weak formulation allows the investor to randomize from the
start, and so the filtration is rich enough in the sense that the law of any
cluster point $(S,X,\hat\gamma)$ can be represented with an adapted process  $\hat\gamma$.
On the other hand, since the enlarged
filtration does not provide any additional (in
comparison with the original filtration) information about the future, the value of
the (concave) utility maximization problem remains the same as in the original setup.
For more details see \cite{CR:17}.

In this paper we do not consider the weak formulation approach, instead we solve the measurability issue by projecting
the limit portfolio $\hat\gamma$ on the investor's filtration
$(\mathcal F_t)_{0\leq t\leq T}$.
We will prove that the projection is well defined and gives an optimal trading strategy.

Let us remark that if we had uniqueness results for the optimal trading strategy then we could prove that
the sequence $(S^{(n)},X^{(n)},\hat\gamma^{(n)})$, $n\in\mathbb N$ converges to
$(S,X,\hat\gamma)$ where $\hat\gamma$ is the unique optimal control and in particular it is adapted to the filtration generated by $X$. Surprisingly, up to date, there are no
results related to the uniqueness of the optimal trading strategy.  Of course, for strictly concave utility we can prove the uniqueness of the optimal terminal wealth but
this does not
imply the uniqueness of the optimal trading strategy (see Remark 6.9 in \cite{S1}). The latter is an interesting question which is left for future research.
\end{rem}
\section{Proof of the Main Results}\label{sec:3}
\subsection{Three Crucial Lemmata}
We start with the following result. Recall that in view of Assumption \ref{asm2.1}(i) $u(x)>-\infty$.
\begin{lem}\label{lem3.1}
The function $u:(0,\infty)\rightarrow \mathbb R\cup \{\infty\}$ is continuous.
Namely, for any $x>0$ we have
$u(x)=\lim_{y\rightarrow x} u(y)$
where a priori the joint value can be equal to $\infty$.
\end{lem}
\begin{proof}
The proof is done by using similar arguments as in Lemma 2.1 in \cite{BDG} for the frictionless case. The only needed property is that
for any $\lambda>0$ and a trading strategy $\gamma$ we have the equality
$V^{\lambda\gamma}_t=\lambda V^{\gamma}_t$, $t\in [0,T]$.
Trivially, this property holds true in our setup.
\end{proof}
Now, we prove the lower bound part in Theorem \ref{thm2.1}.
\begin{lem}\label{lem3.2}
For any $x>0$ we have $$u(x)\leq\lim\inf_{n\rightarrow \infty} u_n(x).$$
\end{lem}
\begin{proof}
\emph{Step 1.}
Let $x>0$.
Without loss of generality (by passing to a subsequence) we assume that $\lim_{n\rightarrow \infty} u_n(x)$ exists.
In view of Lemma \ref{lem3.1} in order to prove the statement, it is sufficient to prove that for any
$\epsilon\in (0, x/3)$ and $\gamma\in\mathcal A(x-3\epsilon)$ we have
\begin{equation}\label{3.1}
\mathbb E_{\mathbb P}[U(x+V^{\gamma}_T-3\epsilon,S)]\leq \lim_{n\rightarrow\infty} u_n(x).
\end{equation}
From the Skorohod representation theorem (Theorem 3 in \cite{D}) and Assumption \ref{asm2.3}(ii), we
can redefine the stochastic processes
$(S^{(n)},X^{(n)})$, $n\in\mathbb N$ and $(S,X)$ on the same probability space such that
\begin{equation}\label{3.2}
(S^{(n)},X^{(n)})\rightarrow (S,X) \ \mbox{a.s.} \ \mbox{on} \ \ \mathbb D([0,T];\mathbb R^{m+1}).
\end{equation}
Choose $\epsilon\in (0,x/3)$ and $\gamma\in\mathcal A(x-3\epsilon)$. We aim to prove (\ref{3.1}).

\emph{Step 2.}
For any $n\in\mathbb N$, let $\Gamma^{(n)}$ be the set of all trading strategies (do not have to satisfy admissibility conditions)
in the $n$--step model.
First, we show that there exists a subsequence $\gamma^{(n)}\in \Gamma^{(n)}$, $n\in\mathbb N$
(for simplicity the subsequence is still
denoted by $n$) such that
\begin{equation}\label{3.2+}
V^{\gamma}_T\leq \lim\inf_{n\rightarrow\infty} V^{\gamma^{(n)}}_T
\end{equation}
and
\begin{equation}\label{3.2++}
 \lim\inf_{n\rightarrow\infty}\inf_{0\leq t\leq T}V^{\gamma^{(n)}}_t\geq 2\epsilon-x.
\end{equation}
To that end, for any $n\in\mathbb N$
define the predictable (with respect to $(\mathcal F_t)_{0\leq t\leq T}$) process
$\bar\gamma^{(n)}=(\bar\gamma^{(n)}_t)_{0\leq t\leq T}$ by
\begin{equation*}
\bar\gamma^{(n)}_t:=\sum_{i=1}^{n-1} \gamma_{\frac{(i-1)T}{n}}\mathbb I_{\frac{iT}{n}\leq t<\frac{(i+1)T}{n}}, \ \ t\in [0,T].
\end{equation*}
We made small shift in time in order to make $\bar\gamma^{(n)}$ predictable.
Clearly,
\begin{equation}\label{3.2+++}
\begin{split}
V^{\gamma}_T&=-\int_{0}^T S_td\gamma_t-\kappa\int_{0}^T S_t|d\gamma_t| \\
&\leq\lim\inf_{n\rightarrow\infty}\left(
-\int_{0}^T S_td\bar{\gamma}^{(n)}_t-\kappa\int_{0}^T S_t|d\bar{\gamma}^{(n)}_t|\right)
\\
&=\lim\inf_{n\rightarrow\infty}V^{\bar\gamma^{(n)}}_T.
\end{split}
\end{equation}
Fix $n\in\mathbb N$, $k=0,1,...,n-1$ and $t\in [\frac{kT}{n},\frac{(k+1)T}{n})$.
 From (\ref{2.0})
 and the simple relations
$$\int_{(i-1)T/n}^{\frac{iT}{n}} d\gamma_u= \int_{\frac{iT}{n}}^{\frac{(i+1)T}{n}}d\bar\gamma^{(n)}_u, \ \ \int_{(i-1)T/n}^{\frac{iT}{n}} |d\gamma_u|\geq \int_{\frac{iT}{n}}^{\frac{(i+1)T}{n}}|d\bar\gamma^{(n)}_u|, \   i=1,...,k,$$
it follows that
\begin{equation*}
\begin{split}
V^{\gamma}_{\frac{kT}{n}}-V^{\bar\gamma^{(n)}}_{\frac{(k+1)T}{n}}&\leq 2|\gamma_{\frac{kT}{n}}|
|S_{\frac{(k+1) T}{n}}-S_{\frac{k T}{n}}| \\
&+2\int_{0}^{\frac{kT}{n}} |d\gamma_u| \sup_{0\leq t_1,t_2\leq \frac{(k+1)T}{n},|t_2-t_1|\leq \frac{2kT}{n} } |S_{t_2}-S_{t_1}|.
\end{split}
\end{equation*}
Moreover, since $\bar\gamma^{(n)}$ is (a random) constant on the interval $[\frac{kT}{n},t]$, then from (\ref{2.0})
$$
V^{\bar\gamma^{(n)}}_{\frac{kT}{n}}-V^{\bar\gamma^{(n)}}_{t}\leq 2|\bar\gamma^{(n)}_{\frac{k T}{n}}| |S_t- S_{\frac{kT}{n}}|.
$$
Thus, for any $n\in\mathbb N$ (notice that $V^{\bar\gamma^{(n)}}_{t}=0$ for $t<T/n$)
$$\inf_{0\leq t\leq T}V^{\bar\gamma^{(n)}}_t\geq \min_{0\leq k\leq n}V^{\gamma}_{\frac{kT}{n}}-
6\int_{0}^{T} |d\gamma_u|\sup_{|t_2-t_1|\leq \frac{2kT}{n} } |S_{t_2}-S_{t_1}|.
$$
Since $\gamma=\mathcal A(x-3\epsilon)$, then we conclude
\begin{equation}\label{3.2++++}
\lim\inf_{n\rightarrow\infty}\inf_{0\leq t\leq T}V^{\bar\gamma^{(n)}}_t\geq \inf_{0\leq t\leq T}V^{\gamma}_t\geq 3\epsilon-x.
\end{equation}
Next, let $\tilde \Gamma$ be the set of all simple integrands of the from
\begin{equation}\label{3.3+}
\tilde\gamma_t=\sum_{i=1}^k \beta_i\mathbb{I}_{t_i\leq t< t_{i+1}},
\end{equation}
where
$k\in\mathbb N$, $0=t_1<t_2<....<t_{k+1}=T$
is a deterministic partition and
\begin{equation}\label{3.4}
\beta_i=\phi_i(X_{a_{i,1}},...,X_{a_{i,m_i}}), \ \ i=1,...,k,
\end{equation}
for a deterministic partition
$0=a_{i,1}<...<a_{i,m_i}=t_{i}$ and a continuous bounded function
$\phi_i:(\mathbb R^m)^{m_i}\rightarrow\mathbb R^d$.

Since the filtration $(\mathcal F_t)_{0\leq t\leq T}$ is generated by $X$ then standard density arguments imply
that any random variable $\mathcal F_{t-}$ measurable can be approximated (with respect to convergence in probability) by random variables of
the form $\phi(X_{a_1},....,X_{a_k})$ where
$0=a_1<...<a_k=t$ is a deterministic partition and $\phi$ is a continuous bounded function.
We obtain that for any $n\in\mathbb N$ the trading strategy $\bar \gamma^{(n)}$ can be approximated by trading strategies
in $\tilde\Gamma$.
This together with (\ref{3.2+++})--(\ref{3.2++++}) gives that
for any $\delta>0$ there exists
$\tilde\gamma\in\tilde\Gamma$ such that
\begin{equation}\label{3.5}
\mathbb P\left(V^{\gamma}_T>\delta+V^{\tilde\gamma}_T\right)<\delta
\end{equation}
and
\begin{equation}\label{3.6}
\mathbb P\left(\inf_{0\leq t\leq T}V^{\tilde\gamma}_t<2\epsilon-x\right)<\delta.
\end{equation}
Let $\tilde \gamma$ be of the form (\ref{3.3+})--(\ref{3.4}).
Define the strategies $\tilde\gamma^{(n)}\in\Gamma_n$, $n\in\mathbb N$ by
$$
\tilde\gamma^{(n)}_t=\sum_{i=1}^k \phi_i(X^{(n)}_{a_{i,1}},...,X^{(n)}_{a_{i,m_i}})\mathbb{I}_{t_i\leq t< t_{i+1}}.
$$
From (\ref{3.2}) and the fact that $\phi_i$, $i=1,...,k$ are continuous we obtain
\begin{equation}\label{3.7}
\lim_{n\rightarrow\infty}\sup_{0\leq t\leq T}|V^{\tilde\gamma}_t-V^{\tilde\gamma^{(n)}}_t|=0 \ \ \mbox{a.s.}
\end{equation}
By applying the Borel--Cantelli lemma and (\ref{3.5})--(\ref{3.7}) we obtain that there exists a subsequence $\gamma^{(n)}\in\Gamma_n$, $n\in\mathbb N$
which
satisfies (\ref{3.2+})--(\ref{3.2++}).

\emph{Step 3.}
Next, we modify the trading strategies $\gamma^{(n)}\in\Gamma_n$, $n\in\mathbb N$ in order to meet the admissibility requirements.
For any $n\in\mathbb N$ define the stopping time
$$\tau_n:=T\wedge \inf\{t:x+V^{\gamma^{(n)}}_t< \epsilon\},$$
and consider the trading strategy
$$\beta^{(n)}_t:=\gamma^{(n)}_t\mathbb {I}_{t<\tau_n}, \  \ t\in [0,T].$$
From (\ref{2.0}) and the definition of $\tau_n$ it follows that for any $n\in\mathbb N$
\begin{equation}\label{3.7+}
V^{\beta^{(n)}}_t=V^{\gamma^{(n)}}_{t}\mathbb{I}_{t<\tau_n}+V^{\gamma^{(n)}}_{\tau_n-}\mathbb{I}_{t\geq\tau_n}\geq\epsilon-x, \ \ t\in[0,T].
\end{equation}
Thus, $\beta^{(n)}\in\mathcal A^{(n)}(x)$, $n\in\mathbb N$.

From (\ref{3.2++}) we have
\begin{equation}\label{3.8}
\mathbb I_{\tau_n=T}\rightarrow 1 \ \  \mbox{a.s.}
\end{equation}
Recall (see the first paragraph after (\ref{2.0}))
that due to the liquidation at the maturity date, the portfolio value processes are continuous at time $T$.
Thus, from (\ref{3.2+}) and (\ref{3.7+})--(\ref{3.8})
 \begin{equation}\label{3.9}
\lim\inf_{n\rightarrow\infty}V^{\beta^{(n)}}_T\geq V^\gamma_T.
\end{equation}

Finally, from Fatou's Lemma, Assumption \ref{asm2.2+}(i), the inequality $x+V^{\beta^{(n)}}_T\geq\epsilon$,
(\ref{3.2}) and (\ref{3.9}) we obtain
$$\lim_{n\rightarrow\infty} u_n(x)\geq
\lim\inf_{n\rightarrow\infty} \mathbb E_{\mathbb P_n}\left[U\left(x+V^{\beta^{(n)}}_T,S^{(n)}\right)\right]
\geq \mathbb E_{\mathbb P}[U(x+V^\gamma_T-3\epsilon,S)]$$
and (\ref{3.1}) follows.
\end{proof}
Next, we prove the following key result.
\begin{lem}\label{lem.tight}
Let $x>0$ and
$\gamma^{(n)}\in\mathcal A^{(n)}(x)$, $n\in\mathbb N$
be a sequence of admissible trading strategies.
The sequence
$(X^{(n)},S^{(n)},\gamma^{(n)})$ is tight on the space
 $\mathbb D([0,T];\mathbb R^{m+1})\times \mathbb D_{MZ}[0,T]$ and so from Prohorov's theorem (see \cite{B})
 it is relatively compact. Moreover, any cluster point is of the form $(X,S,\gamma)$
 and satisfies the following conditional independence property:

 Let $(\mathcal F^{X,\gamma}_t)_{0\leq t\leq T}$
 be the usual filtration generated by $X$
 and $\gamma$. Then for any $t<T$,
 $\mathcal F^{X,\gamma}_t$ and $\mathcal F_T$ are conditionally independent given
 $\mathcal F_t$. As before $(\mathcal F_t)_{0\leq t\leq T}$ is the usual filtration
 generated by $X$.
\end{lem}
\begin{proof}
\emph{Step 1.}
In \cite{MZ} (see Lemma 8 there) the authors proved that for any $c>0$
the set
$$\left\{f: \ f \ \mbox{is} \ \mbox{of} \  \mbox{bounded} \  \mbox{variation} \ \mbox{and} \ \int_{0}^T |df(t)|\leq c\right\}\subset
\mathbb D_{MZ}[0,T]$$ is compact.
Thus, in order to prove tightness (in the Meyer--Zheng topology) of the sequence
$\gamma^{(n)}\in\mathcal A^{(n)}(x)$, $n\in\mathbb N$,
it is sufficient to prove that for any $\delta>0$ there exists $c>0$ and $N\in\mathbb N$ such that
\begin{equation}\label{2+.1}
\mathbb P_n\left(\int_{0}^T|d\gamma^{(n)}_t|> c\right)<2\delta, \ \ \forall n>N.
\end{equation}
Choose $\delta>0$.
Since $S$ is strictly positive then the weak convergence $S^{(n)}\Rightarrow S$ implies that there exists $\hat\delta>0$ and $N\in\mathbb N$ such that
\begin{equation}\label{2+.2}
\mathbb P_n\left(\inf_{0\leq t\leq T}S^{(n)}_t<\hat\delta\right)<\delta, \ \ \forall n>N.
\end{equation}
Next, recall the
processes $M^{(n)}$, $n\in\mathbb N$
and the probability measures $\mathbb Q_n$, $n\in\mathbb N$ given
by Assumption \ref{asm2.2}.
From the inequalities
$$|M^{(n)}_t-S^{(n)}_t|, |M^{(n)}_{t-}-S^{(n)}_{t}|\leq (\kappa-\varepsilon)S^{(n)}_t, \ \ t\in [0,T],$$
 the admissibility property of $\gamma^{(n)}\in\mathcal A^{(n)}(x)$, $n\in\mathbb N$ and the integration by parts formula we get
\begin{equation*}
\begin{split}
0&\leq x+\gamma^{(n)}_tS^{(n)}_t-\int_{0}^t S^{(n)}_ud\gamma^{(n)}_u-\kappa |\gamma^{(n)}_t|S^{(n)}_t -\kappa\int_{0}^t S^{(n)}_u|d\gamma^{(n)}_u|\\
&\leq x
+\gamma^{(n)}_tM^{(n)}_t+(\kappa-\varepsilon)|\gamma^{(n)}_t| S^{(n)}_t-\int_{0}^t M^{(n)}_{u-}d\gamma^{(n)}_u
+(\kappa-\varepsilon)\int_{0}^t S^{(n)}_{u}|d\gamma^{(n)}_u|\\
&-\kappa |\gamma^{(n)}_t|S^{(n)}_t-\kappa\int_{0}^t S^{(n)}_u|d\gamma^{(n)}_u|\\
&\leq x+\int_{0}^t \gamma^{(n)}_{u-}dM^{(n)}_u-\varepsilon \int_{0}^t S^{(n)}_u|d\gamma^{(n)}_u|, \  \ \ \forall t\in [0,T].
\end{split}
\end{equation*}
Thus, for any $n\in\mathbb N$
\begin{equation}\label{con}
\mathbb E_{\mathbb Q_n}\left[\int_{0}^T S^{(n)}_u|d\gamma^{(n)}_u|\right]\leq\frac{x}{\varepsilon}.
\end{equation}
From the Markov inequality
and the fact that the sequence $\mathbb P_n$, $n\in\mathbb N$
is contiguous to the sequence $\mathbb Q_n$, $n\in\mathbb N$ we
conclude that there exists $\hat c>0$ such that
\begin{equation*}
\mathbb P_n\left(\int_{0}^T S^{(n)}_u|d\gamma^{(n)}_u|>\hat c\right)<\delta, \ \ \forall n\in\mathbb N.
\end{equation*}
This together with (\ref{2+.2}) gives that (\ref{2+.1}) holds true for
$c:=\frac{\hat c }{\hat\delta}$ and tightness follows.

\emph{Step 2.}
From Assumption \ref{asm2.3}(ii) we conclude that $(S^{(n)},X^{(n)},\gamma^{(n)})$,
 $n\in\mathbb N$ is tight on the space
 $\mathbb D([0,T];\mathbb R^{m+1})\times \mathbb D_{MZ}[0,T]$ and so from Prohorov's theorem it is relatively compact.
 Moreover, from Assumption \ref{asm2.3}(ii) it follows that any cluster point is of the form $(S,X,\gamma)$, namely the distribution of the first two components
 equals to the distribution of the pair $(S,X)$. The process $\gamma$ is of bounded variation
with right continuous paths.

It remains to establish the conditional independence property. With abuse of notations we denote by $\mathbb E_{\mathbb P}$
the expectation on the corresponding probability space.

Choose $t<T$. We need to show (see Definition 43 in \cite{DM}) that
for any  bounded random variable $Z_1$ which is $\mathcal F_T$ measurable and
a bounded random variable $Z_3$ which is $\mathcal F^{X,\gamma}_t$ measurable we
have the equality
$$\mathbb E_{\mathbb P} (Z_1 Z_3|\mathcal F_t)=\mathbb E_{\mathbb P} (Z_1|\mathcal F_t)\mathbb E_{\mathbb P} ( Z_3|\mathcal F_t). $$
This is equivalent to proving that for any bounded random variable $Z_2$ which is $\mathcal F_t$ measurable we have
\begin{equation*}
\mathbb E_{\mathbb P}(Z_1 Z_2 Z_3)=\mathbb E_{\mathbb P}\left(\mathbb E_{\mathbb P}(Z_1|\mathcal F_t) Z_2 Z_3\right).
\end{equation*}
Since the filtration $(\mathcal F_u)_{0\leq u\leq T}$ is right continuous then the above equality follows from the equality
\begin{equation}\label{2.12}
\mathbb E_{\mathbb P}(Z_1 Z_2 Z_3)=\mathbb E_{\mathbb P}\left(\mathbb E_{\mathbb P}(Z_1|\mathcal F_u) Z_2 Z_3\right), \ \ \forall u>t.
\end{equation}
Standard density arguments yield that without loss of generality we can assume that
 $Z_1=\psi_1(X)$ for a continuous, bounded function $\psi_1:\mathbb D([0,T];\mathbb R^m)\rightarrow\mathbb R$.
Let
$$Y^{(n)}_u:=\mathbb E_{\mathbb P_n}\left(\psi_1(X^{(n)})|\mathcal F^{(n)}_u\right) \ \ \mbox{and} \ \
Y_u:=\mathbb E_{\mathbb P}\left(\psi_1(X)|\mathcal F_u\right), \ \ u\in [0,T].$$
By passing to a subsequence,
we can assume without loss of generality that
$\left(S^{(n)},X^{(n)},\gamma^{(n)}\right)$ converge to $(S,X,\gamma)$.
From Assumption \ref{asm2.3}(ii)--(iii) we obtain that the sequence
$\left(S^{(n)},X^{(n)},Y^{(n)},\gamma^{(n)}\right)$, $n\in\mathbb N$ is tight on the space
$\mathbb D([0,T];\mathbb R^{m+1})\times \mathbb D([0,T];\mathbb R)\times \mathbb D_{MZ}[0,T]$
and so from Prohorov's theorem it is relatively compact.
Moreover, Assumptions \ref{asm2.3}(ii)--(iii) imply that any cluster point is of the form $(S,X,Y,\gamma)$.

From the Skorohod representation theorem (Theorem 3 in \cite{D})
there exists a probability space
which contains a subsequence
$(S^{(n)},X^{(n)},Y^{(n)},\gamma^{(n)})$, $n\in\mathbb N$ and $(S,X,Y,\gamma)$ on the same probability space
such that
\begin{equation}\label{2.10}
(S^{(n)},X^{(n)},Y^{(n)},\gamma^{(n)})\rightarrow (S,X,Y,\gamma) \ \ \mbox{a.s. }
\end{equation}
on the space  $\mathbb D([0,T];\mathbb R^{m+1})\times \mathbb D([0,T];\mathbb R)\times \mathbb D_{MZ}[0,T]$.

Next, from the bounded convergence theorem we have
 $$\lim_{n\rightarrow\infty} \mathbb E\left(\int_{0}^T \min(1,|\gamma^{(n)}_u-\gamma_u|)du\right)=0,$$
 where $\mathbb E$ denotes the expectation on the common probability space.
From Fubini's theorem we conclude that there exists a subset $I\subset [0,T]$ of a full Lebesgue measure such that
\begin{equation}\label{2.11}
\gamma_u=\lim_{n\rightarrow \infty}\gamma^{(n)}_u, \ \ \forall u\in I,
\end{equation}
where the limit is in probability.

Choose $u>t$ in (\ref{2.12}). Again, standard density arguments imply that without loss of generality we can assume that
$Z_2,Z_3$ in (\ref{2.12}) are of the form: $Z_2=\psi_2(X_{t_1},...,X_{t_k})$ for some $k\in\mathbb N$,
 $t_1,...,t_k\in I\cap [0,u]$ and a continuous bounded function
 $\psi_2:(\mathbb R^m)^k\rightarrow\mathbb R$ and
 $Z_3=\psi_3(X_{t_1},...,X_{t_k},\gamma_{t_1},...,\gamma_{t_k})$
 for a continuous bounded function
  $\psi_3:(\mathbb R^m)^k\times \mathbb R^k\rightarrow\mathbb R$.

From the bounded convergence theorem, the fact that $\gamma^{(n)}$ is $\mathcal F^{(n)}$ adapted and
(\ref{2.10})--(\ref{2.11}) we obtain
\begin{equation*}
\begin{split}
& \mathbb E_{\mathbb P}(Z_1 Z_2 Z_3)\\
&=\lim_{n\rightarrow\infty} \mathbb E_{\mathbb P_n}\left(\psi_1(X^n) \psi_2(X^{(n)}_{t_1},...,X^{(n)}_{t_k})
\psi_3(X^{(n)}_{t_1},...,X^{(n)}_{t_k},\gamma^{(n)}_{t_1},...,\gamma^{(n)}_{t_k})\right)\\
&=\lim_{n\rightarrow\infty} \mathbb E_{\mathbb P_n}\left(Y^{(n)}_u \psi_2(X^{(n)}_{t_1},...,X^{(n)}_{t_k})
\psi_3(X^{(n)}_{t_1},...,X^{(n)}_{t_k},\gamma^{(n)}_{t_1},...,\gamma^{(n)}_{t_k})\right)\\
&=\mathbb E_{\mathbb P}(Y_u Z_2 Z_3).
\end{split}
\end{equation*}
This completes the proof of (\ref{2.12}).
\end{proof}

Now, we are ready to complete the proof of the main results.

\subsection{Completion of the Proof of Theorems \ref{thm2.1}--\ref{thm2.2}.}
\begin{proof}
We will prove these two results together.

\emph{Step 1.}
Let $x>0$ and let $\hat\gamma^{(n)}\in\mathcal{A}^{(n)}(x)$, $n\in\mathbb N$ be a sequence of portfolios
 which satisfy (\ref{optimal}).
 By passing to a subsequence we assume without loss of generality that
 $\lim_{n\rightarrow\infty}u_n(x)$
 exists. From (\ref{optimal}) $\lim_{n\rightarrow\infty} \mathbb E_{\mathbb P_n}\left[U\left(x+V^{\hat\gamma^{(n)}}_T,S^{(n)}\right)\right]$
 exists as well.

 From Lemma \ref{lem.tight} we obtain that
 $(S^{(n)},X^{(n)},\hat\gamma^{(n)})$,
 $n\in\mathbb N$ is tight and
any cluster point
is of the form $(S,X,\hat\gamma)$. Obviously $\hat\gamma$ is a right continuous process of bounded variation. Moreover,
since $\hat\gamma^{(n)}_T=0$ for all $n$,
$\hat \gamma_T=0$ as well.
Thus, we define $(V^{\hat\gamma}_t)_{0\leq t\leq T}$ by (\ref{2.0}).

Let us prove the admissibility condition
\begin{equation}\label{2.23}
x+V^{\hat\gamma}_t\geq 0, \ \ \forall t\in [0,T],
\end{equation}
and the inequality
\begin{equation}\label{2.24}
\mathbb E_{\mathbb P}[U(x+V^{\hat\gamma}_T,S)]\geq \lim_{n\rightarrow\infty} u_n(x).
\end{equation}
From (\ref{2+.1})
it follows that the sequence
$$\left(S^{(n)},X^{(n)},\hat\gamma^{(n)},\int_{0}^T |d\hat\gamma^{(n)}_t|\right), \ \ n\in\mathbb N,$$
 is tight on the space
$\mathbb D([0,T];\mathbb R^{m+1})\times \mathbb D_{MZ}[0,T]\times\mathbb R$.

Thus, by passing to a further subsequence and applying the Skorohod representation theorem we obtain that there exists a common probability space
where we have the almost surely convergence
\begin{equation}\label{2.23+}
\left(S^{(n)},X^{(n)},\hat\gamma^{(n)},\int_{0}^T |d\hat\gamma^{(n)}_t|\right)\rightarrow (S,X,\hat\gamma,\eta) \ \ \mbox{a.s.}
\end{equation}
for some random variable $\eta<\infty$. We conclude that
\begin{equation}\label{2.25}
\sup_{n\in\mathbb N} \int_{0}^T |d\hat\gamma^{(n)}_t|<\infty \ \ \mbox{a.s.}
\end{equation}
Next, using the same arguments as before (\ref{2.11}) gives that there exists a subset $I\subset [0,T]$ of a full Lebesgue measure such that
$\hat\gamma^{(n)}_u\rightarrow\hat\gamma_u$ in probability for all $u\in I$.

From a diagonalization argument, it follows that there exists a countable dense set $\hat I$ and a subsequence
$(S^{(n)}, X^{(n)}, \hat \gamma^{(n)})$, $n\in\mathbb N$,
such that
\begin{equation}\label{2.26}
\hat\gamma_u=\lim_{n\rightarrow \infty} \hat \gamma^{(n)}_u \ \ \mbox{a.s.} \ \ \forall u\in \hat I.
\end{equation}
Since $\hat\gamma_T=0$ and $\hat\gamma^{(n)}_T=0$, $n\in\mathbb N$ then without loss of generality we assume that
$T\in\hat I$.

Let $t\in\hat I$.
From Theorem 12.16 in \cite{MP} and (\ref{2.23+})--(\ref{2.26})
\begin{equation}\label{4.refer1}
\int_{0}^t S_u d\hat\gamma_u=\lim_{n\rightarrow \infty}\int_{0}^t S^{(n)}_u d\hat\gamma^{(n)}_u \ \ \mbox{a.s.}
\end{equation}
Next, choose
a partition $\{0=a_0<a_1<...<a_k=t\}\subset\hat I$.
From (\ref{2.23+})--(\ref{2.26})
\begin{equation*}
\begin{split}
\lim\inf_{n\rightarrow\infty} \int_{0}^t S^{(n)}_u |d\hat\gamma^{(n)}_u|
&= \lim\inf_{n\rightarrow\infty} \int_{0}^t S_u |d\hat\gamma^{(n)}_u|\\
&\geq \lim\inf_{n\rightarrow\infty}\sum_{i=0}^{k-1}\min_{a_i\leq u\leq a_{i+1}} S_u |\hat\gamma^{(n)}_{a_{i+1}}-\hat\gamma^{(n)}_{a_i}|\\
&=\sum_{i=0}^{k-1}\min_{a_i\leq u\leq a_{i+1}} S_u |\hat\gamma_{a_{i+1}}-\hat\gamma_{a_i}|   .
\end{split}
\end{equation*}
By taking the mesh of the
partition to zero we conclude
\begin{equation}\label{4.refer2}
\lim\inf_{n\rightarrow\infty} \int_{0}^t S^{(n)}_u |d\hat\gamma^{(n)}_u|\geq \int_{0}^t S_u |d\hat\gamma_u|  \ \ \mbox{a.s.}
\end{equation}
From (\ref{2.26})--(\ref{4.refer2})
\begin{equation}\label{2.27}
V^{\hat\gamma}_t\geq\lim\sup_{n\rightarrow\infty} V^{\hat\gamma^{(n)}}_t\geq -x \ \ \mbox{a.s.} \ \ \forall t\in \hat I.
\end{equation}
Since $\hat I$ is a dense set which contains $T$ and
$(V^{\hat\gamma}_t)_{0\leq t\leq T}$ is RCLL we obtain (\ref{2.23}).
From Assumption \ref{asm2.2+}(ii), (\ref{optimal}) and (\ref{2.27}) (for $t=T$) we conclude (\ref{2.24}).

\emph{Step 2.}
From the conditional independence property that was proved in Lemma \ref{lem.tight}
it follows that any martingale with respect to the filtration $(\mathcal F_t)_{0\leq t\leq T}$
is also a martingale with respect to the filtration $(\mathcal F^{X,\hat\gamma}_t)_{0\leq t\leq T}$. Thus, we redefine
the probability measure $\mathbb Q$ and the $\mathbb Q$ local--martingale $M$ (from Assumption \ref{asm2.2}) on the common probability space.

Since $M$ is a $\mathbb Q$ local--martingale with respect to the filtration $(\mathcal F^{X,\hat\gamma}_t)_{0\leq t\leq T}$,
then by
the same arguments as before (\ref{con}), we obtain that
(\ref{2.23}) implies
$
\mathbb E_{\mathbb Q}\left[\int_{0}^T S_u|d\hat\gamma_u|\right]\leq\frac{x}{\varepsilon}.
$
This together with the inequality $\inf_{0\leq u\leq T}S_u>0$ a.s. and
the fact that $\frac{d\mathbb Q}{d\mathbb P}$
is $\mathcal F_T$ measurable
gives
\begin{equation}\label{2.51}
\mathbb E_{\mathbb P}\left(\int_{0}^T |d\hat\gamma_u|\bigg|\mathcal F_T\right)<\infty \ \ \mbox{a.s.}
\end{equation}
As usual for any non--negative random variable $\xi\geq 0$ and a $\sigma$--algebra
$\mathcal G$ we define the conditional expectation $\mathbb E(\xi|\mathcal G)$ as a non--negative extended random variable.
For a general random variable $\xi$ we use the convention
$\mathbb E(\xi|\mathcal G):=\mathbb E(\xi^{+}|\mathcal G)-\mathbb E(\xi^{-}|\mathcal G)$
a.s. on the set where $\mathbb E(\xi^{+}|\mathcal G),\mathbb E(\xi^{-}|\mathcal G)<\infty$
and $\mathbb E(\xi|\mathcal G):=\infty$ on the complement.

Let $\hat\gamma_t:=\hat\gamma^{+}_t-\hat\gamma^{-}_t$, $t\in [0,T]$
be the Jordan
decomposition into a positive variation
$(\hat\gamma^{+}_t)_{0\leq t\leq T}$ and a negative variation
$(\hat\gamma^{-}_t)_{0\leq t\leq T}$.

Define the processes
$$\hat\beta^{+}_t:=\mathbb E_{\mathbb P}\left(\hat\gamma^{+}_t|\mathcal F_T\right), \ \
\hat\beta^{-}_t:=\mathbb E_{\mathbb P}\left(\hat\gamma^{-}_t|\mathcal F_T\right), \ \ t\in [0,T].
$$
From (\ref{2.51}) we obtain that
$(\hat\beta^{+}_t)_{0\leq t\leq T}$ and $(\hat\beta^{-}_t)_{0\leq t\leq T}$ are non--decreasing and non--negative processes
which satisfy  $\hat\beta^{+}_T,\hat\beta^{-}_T<\infty$ a.s.
Moreover, from the dominated convergence theorem (for conditional expectation), the fact that $\hat\gamma$ is right continuous
and (\ref{2.51}) it follows that
$\hat\beta^{+},\hat\beta^{-}$
are right continuous as well.

Next, fix $t\in [0,T]$ and $n\in\mathbb N$. From \cite{DM} (see Chapter 2, Theorem 45)
and the conditional independence property that was proved in Lemma \ref{lem.tight}
 $$
\mathbb E_{\mathbb P}\left(\hat\gamma^{+}_t\wedge n|\mathcal F_t\right)= \mathbb E_{\mathbb P}\left(\hat\gamma^{+}_t\wedge n|\mathcal F_T\right).$$
By taking $n\rightarrow\infty$ we get
\begin{equation}\label{2.i}
\mathbb E_{\mathbb P}\left(\hat\gamma^{+}_t|\mathcal F_t\right)= \hat{\beta}^{+}_t, \ \ \forall t\in [0,T].
\end{equation}
Similarly,
\begin{equation}\label{2.ii}
\mathbb E_{\mathbb P}\left(\hat\gamma^{-}_t|\mathcal F_t\right)= \hat{\beta}^{-}_t, \ \ \forall t\in [0,T].
\end{equation}
Thus, from (\ref{portfolio})
\begin{equation}\label{2.52}
\hat\gamma^{\mathcal F}_t=\hat\beta^{+}_t-\hat\beta^{-}_t=\mathbb E_{\mathbb P}(\hat\gamma_t|\mathcal F_T), \ \ \forall t\in [0,T].
\end{equation}
We conclude that $(\hat\gamma^{\mathcal F}_t)_{0\leq t\leq T}$ is a right continuous process of bounded variation.
Clearly, (\ref{portfolio}) implies that $(\hat\gamma^{\mathcal F}_t)_{0\leq t\leq T}$
is adapted to the filtration
$(\mathcal F_t)_{0\leq t\leq T}$.

From (\ref{2.51})
we have
$
\mathbb E_{\mathbb P}\left(\sup_{0\leq t\leq T}S_t\int_{0}^T |d\hat\gamma_u|\big|\mathcal F_T\right)<\infty $ a.s.
Hence, from the dominated convergence theorem and (\ref{2.i})--(\ref{2.52})
\begin{equation*}
\int_{0}^t S_u d\hat\gamma^{\mathcal F}_u=\mathbb E_{\mathbb P}\left( \int_{0}^t S_u d\hat\gamma_u\big|\mathcal F_T\right), \ \ \forall t\in [0,T]
\end{equation*}
and
\begin{equation*}
\begin{split}
\int_{0}^t S_u |d\hat\gamma^{\mathcal F}_u|&\leq \int_{0}^t S_u d\hat\beta^{+}_u+\int_{0}^t S_u d\hat\beta^{-}_u\nonumber\\
&=\mathbb E_{\mathbb P}\left( \int_{0}^t S_u d\hat\gamma^{+}_u\big|\mathcal F_T\right)+
\mathbb E_{\mathbb P}\left( \int_{0}^t S_u d\hat\gamma^{-}_u\big|\mathcal F_T\right)\nonumber\\
&= \mathbb E_{\mathbb P}\left(\int_{0}^t S_u |d\hat\gamma_u|\big|\mathcal F_T\right), \ \ \forall t\in [0,T].
\end{split}
\end{equation*}
This together with (\ref{2.0}) and the simple relations
$$\hat\gamma^{\mathcal F}_t S_t=\mathbb E_{\mathbb P}(\hat\gamma_t S_t\big|\mathcal F_T),
\ \ |\hat\gamma^{\mathcal F}_t S_t|\leq\mathbb E_{\mathbb P}(|\hat\gamma_t| S_t\big|\mathcal F_T), \ \ \forall t\in [0,T]$$
gives
\begin{equation}\label{2.53}
V^{\hat\gamma^{\mathcal F}}_t\geq \mathbb E_{\mathbb P}(V^{\hat{\gamma}}_t|\mathcal F_T), \ \ \forall t\in [0,T],
\end{equation}
and so from (\ref{2.23}) $\hat\gamma^{\mathcal F}\in\mathcal A(x)$.

Thus, from the Jensen inequality,
Lemma \ref{lem3.2},
(\ref{2.24}) and (\ref{2.53})
$$
\lim_{n\rightarrow\infty} u_n(x)\geq u(x)\geq \mathbb E_{\mathbb P}[U(x+V^{\hat\gamma^{\mathcal F}}_T,S)]
\geq
\mathbb E_{\mathbb P}[U(x+V^{\hat\gamma}_T,S)]\geq\lim_{n\rightarrow\infty} u_n(x),$$
and the proof is completed.

\end{proof}
\section{Example: Transaction Costs Make Things Converge}\label{sec:5}
Consider a random utility which corresponds to shortfall risk minimization
for a call option with strike price $K>0$. Namely, we set
$$U(v,s):=-\left(\left(s_T-K\right)^{+}-v\right)^{+}, \ \ \forall (v,s)\in [0,\infty) \times C[0,T].$$

Next, let $\xi_i=\pm 1$, $i\in\mathbb N$ be i.i.d. and symmetric. For any $n\in\mathbb N$
define the scaled random walks
$X^{{n,1}}_t, X^{{n,2}}_t$, $t\in [0,T]$ by
\begin{equation*}
\begin{split}
X^{{n,1}}_t&:=\sqrt\frac{T}{n}\sum_{i=1}^{[nt/T]} \xi_i, \\
X^{n,2}_t&:=\sqrt\frac{T}{n}\sum_{i=1}^{[nt/T]} \prod_{j=1}^i \xi_j,
\end{split}
\end{equation*}
where $[\cdot]$ is the integer part of $\cdot$ and we set $\sum_{i=1}^0\equiv 0$.

Let $X^{(n)}:=(X^{{n,1}},X^{{n,2}})$
and let $(\mathcal F^{(n)})_{0\leq t\leq T}$ be the usual filtration generated by $X^{(n)}$. Observe that
$X^{(n)}$ is a martingale.
From Lemma 3.1 in \cite{BDG} it follows that we have the weak convergence
$X^{(n)}\Rightarrow X$
where $X=(X^{(1)},X^{(2)})$ is a standard two dimensional Brownain motion.
From
Corollary 6 in \cite{JS} we conclude the extended weak convergence
\begin{equation}\label{5.1}
X^{(n)}\Rrightarrow X.
\end{equation}
We remark that although \cite{JS} deals only with real valued processes,
the extension of the results there to the multidimensional case is straightforward.

Now, we introduce the financial markets.
For any $n\in\mathbb N$ define the discrete time stochastic processes
$\{\tilde\nu^{(n)}_k\}_{k=0}^n$ and $\{\tilde S^{(n)}_k\}_{k=0}^n$ by
\begin{equation*}
\begin{split}
\tilde\nu^{(n)}_k&:=\prod_{i=1}^k \left(1+\sqrt\frac{T}{n}\xi_i\right),\\
\tilde S^{(n)}_k&:=\prod_{i=1}^k \left(1+\min\left(\nu^{(n)}_{\frac{(i-1)T}{n}},\ln n\right)\sqrt\frac{T}{n}\prod_{j=1}^i\xi_j\right),
\end{split}
\end{equation*}
we set $\prod_{i=1}^0\equiv 1$.
Let $S^{(n)}=(S^{(n)}_t)_{0\leq t\leq t}$, $n\in\mathbb N$ be the following linear interpolation
\begin{equation}\label{5.0}
S^{(n)}_t:=(\left([nt/T]+1\right)T-nt)\tilde S^{(n)}_{[nt/T]-1}+\left(nt-[nt/T]T\right)  \tilde S^{(n)}_{[nt/T]},
\end{equation}
where we set $\tilde S^{(n)}_{-1}\equiv 1$.
We take a shift of one time period in order
to make $S^{(n)}$ adapted to $\mathcal F^{(n)}$.
For any $n$ define the process $(\nu^{(n)}_t)_{0\leq t\leq T}$by
$\nu^{(n)}_t:=\tilde\nu^{(n)}_{[nt/T]}$.

Using the same arguments as in Example 3.3 in \cite{BDG} we obtain the weak convergence
\begin{equation}\label{5.2}
(\nu^{(n)}, S^{(n)},X^{(n)})\Rightarrow (\nu,S,X),
\end{equation}
where $(\nu,S)$ is the (unique strong) solution of the SDE (Hull and White model)
\begin{equation}\label{5.2+}
\begin{split}
d\nu_t&=\nu_t dX^{(1)}_t, \ \ \nu_0=1, \\
dS_t&=\nu_t S_t dX^{(2)}_t, \ \ S_0=1.
\end{split}
\end{equation}

Next, we verify
Assumptions \ref{asm2.1}-\ref{asm2.3}.
Clearly, $\mathbb E_{\mathbb P}S_T\leq S_0=1$. Thus Assumption \ref{asm2.1}(i) holds true.
From Remark \ref{rem.short} it follows that
Assumption \ref{asm2.1}(ii) holds true as well.

Let $\kappa>0$ be the rate of the transaction costs.
Observe that $\{\tilde S^{(n)}_k\}_{k=0}^n$ is a martingale. Thus, for sufficiently
large $n$ we obtain
that the martingale $M^{(n)}=(M^{(n)}_t)_{0\leq t\leq T}$ given by
$M^{(n)}_t:=\tilde S^{(n)}_{[nt/T]}$
satisfies
$|M^{(n)}-S^{(n)}|\leq \frac{\kappa}{2} S^{(n)}.$
Hence, Assumptions \ref{asm2.2} holds true
for $\mathbb Q_n:=\mathbb P_n$. For the limit model we just take
$M:=S$ and $\mathbb Q:=\mathbb P$.

Similar arguments as in Example 3.3 in \cite{BDG} yield that the random variables
$\{S^{(n)}_T\}_{n\in\mathbb N}$ are uniformly integrable. This gives Assumption \ref{asm2.2+}(i).
Since $U^{+}=0$ then Assumption \ref{asm2.2+}(ii) is trivial.
Finally, Assumption \ref{asm2.3} follows from (\ref{5.1}) and (\ref{5.2}).

From Theorem \ref{thm2.1} it follows that
in the presence of proportional transaction costs, the shortfall risks in the models given by (\ref{5.0})
converge to the shortfall risk in the Hull--White stochastic volatility model given by (\ref{5.2+}).

On the other hand, for the frictionless case i.e. $\kappa=0$ there is no convergence.
In fact, the models given by (\ref{5.0})
are not arbitrage free. Observe that for any interval of the form
$[kT/n, (k+1) T/n]$, $k\geq 1$, immediately after time $kT/n$ the investor knows all the stock prices in this interval.
Thus, we have an obvious arbitrage, which makes the shortfall risk equal to zero for any initial capital.
Clearly, this is not the case for the limit model.

Moreover, assume that for the $n$--step model
we restrict the trading times to the set
$0,T/n,2  T/n,...,T$.
We notice that for $k=1,...,n-1$
the conditional support
$supp \left(S^{(n)}_{\frac{(k+1)T}{n}}|S^{(n)}_{\frac{T}{n}},...,S^{(n)}_{\frac{kT}{n}}\right)$
consists of exactly two points.
Namely, the discretized market which is active at times $0,T/n,2  T/n,...,T$ is complete, in particular arbitrage free.
Still,
even with this restriction
(in the frictionless setup) the shortfall risk in the Hull--White model is
strictly bigger than the limit of the shortfall risk
in the models given by $S^{(n)}$, $n\in\mathbb N.$
This fact was established in Example 3.3 in \cite{BDG}.


\begin{thebibliography}{}

\bibitem{A}
D. Aldous,
{\em Weak convergence of stochastic processes for processes viewed in the
strasbourg manner},
unpublished manuscript, (1981).


\bibitem{B}
P. Billingsley,
{\em Convergence of probability measures,}
Wiley Series in Probability and
Statistics: Probability and Statistics. John Wiley \& Sons, Inc., New York, second edition, (1999).


\bibitem{Bruno:02}
B. Bouchard,
{\em Utility maximization on the real line under proportional transaction costs,}
Finance and Stochastics, {\bf 6, } 495--516, (2002).



\bibitem{BBBE1}
J. Backhoff--Veraguas, D. Bartl, M. Beiglböck and M. Eder,
{\em All Adapted Topologies are Equal,}
arxiv: 1905.00368, (2019).

\bibitem{BBBE2}
J. Backhoff--Veraguas, D. Bartl, M. Beiglböck and M. Eder,
{\em Adapted Wasserstein Distances and Stability in Mathematical Finance,}
to appear in Finance and Stochastic.
arxiv: 1901.07450, (2019).



\bibitem{BDG}
E. Bayraktar, Y. Dolinsky and J. Guo,
{\em Continuity of Utility Maximization under Weak Convergence,}
to appear in Mathematics and Financial Economics. arxiv: 1811.01420, (2020).



\bibitem{BY:18}
E. Bayraktar and X. Yu,
{\em On the Market Viability under Proportional Transaction Costs,}
Mathematical Finance, {\bf 28}, 800--838, (2018).


\bibitem{CK:96}
J. Cvitanić and I. Karatzas,
{\em Hedging and portfolio optimization under transaction costs: A martingale approach,}
Mathematical Finance, {\bf 6}, 133--165, (1996).


\bibitem{CO:11}
L. Campi a d M. Owen,
{\em Multivariate Utility Maximization with Proportional Transaction Costs,} Finance and Stochastics, {\bf 15,} 461--499, (2011).

\bibitem{CR:17}
H. N. Chau and M. Rásonyi,
{\em Skorohod's Representation Theorem and Optimal Strategies for Markets with Frictions,}
SIAM Journal on Control and Optimization, {\bf 55,}
 3592--3608, 2017.


\bibitem{CS:16}
Ch. Czichowsky and W. Schachermayer,
{\em Duality Theory for Portfolio Optimisation under Transaction Costs,}
Annals of Applied Probability, {\bf 26}, 1888--1941, (2016).


\bibitem{D}
R.M. Dudley,
{\em Distances of probability measures and random variables,}
Ann. Math. Statist, {\bf 39}, 1563--1572, (1968).

\bibitem{DL:91}
B. Dumas and E. Luciano,
{\em An exact solution to a dynamic portfolio choice
problem under transaction costs,}
Jounral of Finance, {\bf 46}, 577--595, (1991).

\bibitem{DM}
C. Dellacherie and P.A. Meyer,
{\em Probabilités et Potentiel,}
 volumes A et B. Hermann,
Paris, North Holland, Amsterdam, (1978).


\bibitem{DN:90}
M. H. A. Davis and A. R. Norman,
{\em Portfolio selection with transaction costs,}
Math. Oper. Research, {\bf 15}, 676--713, (1990).


\bibitem{DPT:01}
G. Deelstra, H. Pham and N. Touzi,
{\em Dual formulation of the utility maximization problem under transaction costs,}
Annals of Applied Probability, {\bf 11,} 1353--1383, (2001).



\bibitem{G:02}
P. Guasoni
{\em Optimal investment with transaction costs and without semimartingales,}
   Annals of Applied Probability, {\bf 12}, 1227--1246, (2002).


\bibitem{GGKS:14}
S. Gerhold, P. Guasoni, J. Muhle-Karbe and W. Schachermayer,
{\em Transaction Costs, Trading Volume, and the Liquidity Premium,}
Finance and Stochastics, {\bf 18}, 1--37, (2014).



\bibitem{GLR}
P. Guasoni, E. Lépinette and M. Rásonyi,
{\em The fundamental theorem of asset pricing
under transaction costs,} Finance and Stochastics,
{\bf 16}, 741--777, (2012).



\bibitem{JS} A. Jakubowski and L. Slominski,
{\em Extended convergence to continuous in probability processes with independent increments,}
Probab. Theory Relat. Fields. {\bf 72}, 55--82, (1986).


\bibitem{KK:10}
J. Kallsen and J. Muhle--Karbe,
{\em On using shadow prices in portfolio optimization with transaction costs,}
Annals of Applied Probability, {\bf 20}, 1341--1358, (2010).





\bibitem{MC:76}
M. Magill and G. Constantinides,
{\em Portfolio selection with transaction costs,}
Journal of
economic theory,
 {\bf 13}, 245--263, (1976).



\bibitem{MP} B.C. Morrey and H.M. Protter,
{\em A First Course in Real Analysis,}
Springer--Verlag, (1991).


\bibitem{MZ}
P. Meyer and W. Zheng,
{\em Tightness criteria for laws of semimartingales,}
Ann. Inst.
Henri Poincare, {\bf 20}, 353--372, (1984).



\bibitem{S}
W. Schachermayer,
{\em Admissible Trading Strategies under Transaction Costs,}
Seminaire de Probabilite XLVI, Lecture Notes in Mathematics 2123, 317--331, (2015).

\bibitem{S1}
W. Schachermayer,
{\em The Asymptotic Theory of Transaction Costs,}
Zürich Lectures in Advanced Mathematics (165 pages), European Mathematical Society, (2017).


\bibitem{SS:94}
S. Shreve and H.M Soner,
{\em Optimal investment and consumption with transaction costs,}
Annals of Applied Probability, {\bf 3}, 609--692, (1994).


\bibitem{TKA:88}
M. Taksar, M.J Klass, and D. Assaf,
{\em
Diffusion model for optimal portfolio selection in the
presence of brokerage fees,}
Math. Oper. Research, {\bf 13}, 277--294, (1988).



\end{thebibliography}
\end{document}